\documentclass{siamart171218}

\usepackage{graphicx}
\usepackage{epstopdf}
\usepackage{algorithmic}
\ifpdf
  \DeclareGraphicsExtensions{.eps,.pdf,.png,.jpg}
\else
  \DeclareGraphicsExtensions{.eps}
\fi

\newsiamremark{remark}{Remark}
\newsiamremark{hypothesis}{Hypothesis}
\crefname{hypothesis}{Hypothesis}{Hypotheses}
\newsiamthm{claim}{Claim}

\headers{Complex SD for Time Series}{M. Sourisseau, Y. G. Wang, H.-T. Wu, W.-H. Yu}

\title{Optimization-based quasi-uniform complex spherical $(k,l)$-design and generalized multitaper for time series analysis\thanks{\today.
\funding{The first author acknowledges support
from the Australian Research Council under Discovery Project DP180100506.
The last author acknowledges support by Taiwan MOST Grant 107-2115-M-008-010-MY2.
This work is supported by the National Science Foundation under Grant No.~DMS-1439786 while the first and last authors were in residence at the Institute for Computational and Experimental Research in Mathematics in Providence, RI, during the Point configurations in Geometry, Physics and Computer Science Program.
}}}

\author{Matt Sourisseau\thanks{Department of Mathematics, University of Toronto, Toronto, ON, Canada
  (\email{a}).}
\and Yu Guang Wang\thanks{School of Mathematics and Statistics, University of New South Wales, Sydney, NSW, Australia; ICERM, Brown University, Providence, RI, USA 
  (\email{yuguang.wang@unsw.edu.au}).}
\and Hau-Tieng Wu\thanks{Department of Mathematics, Department of Statistical Science, Duke University, Durham, NC, USA; Mathematics Division, National Center for Theoretical Sciences, Taipei, Taiwan 
  (\email{a}).}
\and Wei-Hsuan Yu\thanks{Department of Mathematics, National Central University, Taoyuan, Taiwan; ICERM, Brown University, Providence, RI, USA 
  (\email{hauwu@math.duke.edu}).}}

\usepackage{color}
\usepackage{graphicx} 
\usepackage{subfigure}
\usepackage{graphics,url} 
\usepackage{amssymb}
\usepackage{amsmath}
\usepackage{amsfonts}
\usepackage{amsopn}
\usepackage{mathrsfs}
\usepackage{booktabs}
\usepackage{multirow}

\newtheorem{example}[theorem]{Example}

\newcommand{\dd}[1]{\mathrm{d}#1}

\newcommand{\rsd}{
real spherical $t$-design}

\newcommand{\cpsd}{
complex spherical design}

\newcommand{\tcsd}{
triangular complex spherical $t$-design}

\newcommand{\scsd}{
square complex spherical $t$-design}

\newcommand{\Rd}[1][d]{ 
    \mathbb{R}^{#1}
}

\newcommand{\cd}[1][d]{ 
    \mathbb{C}^{#1}
}

\newcommand{\csd}[1][d]{ 
    \Omega^{#1}
}

\newcommand{\cu}{ 
    \PT{u}
}

\newcommand{\cv}{ 
    \PT{v}
}

\newcommand{\Nz}{ 
    \mathbb{N}_{0}
}

\newcommand{\imu}{ 
    {\mathrm{i}}
}

\newcommand{\PT}[1]{{\bf #1}} 

\newcommand{\conj}[1]{ 
 \overline{#1}
}

\newcommand{\spdk}[1][t]{   
\varphi_{#1}
}

\newcommand{\spdka}[1][t]{   
\widetilde{\varphi}_{#1}
}

\newcommand{\br}{   
\boldsymbol{r}
}

\newcommand{\bD}{   
\boldsymbol{D}
}

\newcommand{\sfgrad}[1][]{ 
	\nabla_{*}
}

\newcommand{\sfcurl}[1][]{ 
	\mathbf{L}
}

\newcommand{\Jcb}[2][\alpha,\beta]{
    P^{(#1)}_{#2}
}

\newcommand{\Legen}[2][d+1]{   
{P}_{#2}^{(#1)}
}

\newcommand{\Jw}[1][\alpha,\beta]{ 
w_{#1}
}

\newcommand{\distr} { 
  {\rm dist}_{r}  }
  
\newcommand{\sepr} { 
  \delta  }
  
\newcommand{\coverr} { 
  h }
  
\newcommand{\mratr} { 
  \rho  }
  
\newcommand{\distc} { 
  {\rm dist}_{\rm c}}
  
\newcommand{\sepc} { 
  \delta^{\rm c}  }
  
\newcommand{\coverc} { 
  h^{\rm c}  }
  
\newcommand{\mratc} { 
  \rho^{\rm c}  }
  
\newcommand{\re}{ 
  {\rm Re}
  }
  
\newcommand{\im}{ 
  {\rm Im}
  }

\newcommand{\sph}[1]{ 
    \mathbb{S}^{#1}
}

\newcommand{\csdm}[1][]{ 
    \mu_{#1}
}

\newcommand{\imat}[1][2]{ 
    I
}

\newcommand{\Lpc}[2][d]{ 
L_{#2}(\csd[#1])
}

\newcommand{\sobH}[1][s]{ 
\mathbb{H}^{#1}
}

\newcommand{\Lpw}[2][\Jw]{ 
\mathbb{L}_{#2}(#1)
}

\newcommand{\shY}[1][\ell,m]{ 
	Y_{#1}
}

\newcommand{\innerh}[2]{ 
\left\langle #1,#2 \right\rangle
}

\newcommand{\innerlc}[2][\Lpc{2}]{ 
\left\langle #2 \right\rangle_{#1}
}

\newcommand{\InnerLGb}[2][{\Jw[r-\frac{1}{2},r-\frac{1}{2}]}]{ 
\left(#2\right)_{\Lpw[{#1}]{2}}
}

\newcommand{\intc}[1]{    
\:\mathrm{d}\mu(\PT{#1})
}

\newcommand{\IntD}[1]{ 
\:\mathrm{d}{#1}
}

\newcommand{\Diff}[2][t]{ 
\ifthenelse{\equal{#2}{1}}{\frac{\mathrm{d}}{\mathrm{d}#1}}{
\left(\frac{\mathrm{d}}{\mathrm{d}#1}\right)^{#2}}
}

\newcommand{\expect}[1]{ 
\mathbb{E}\left[#1\right]
}

\newcommand{\wce}[1]{ 
    \mathrm{wce}(#1)
}

\newcommand{\bigo}[2] { 
    \mathcal{O}_{#1}\left(#2\right)
}

\newcommand{\C}{\mathbb{C}}
\newcommand{\R}{\mathbb{R}}

\newcommand{\bx}{{\bf x}}
\newcommand{\by}{{\bf y}}
\newcommand{\bz}{{\bf z}}

\ifpdf
\hypersetup{
  pdftitle={},
  pdfauthor={}
}
\fi

\headers{Complex Triangular Spherical Designs}{Y. G. Wang, R. S. Womersley, H.-T. Wu, W.-H. Yu}

\title{Numerical computation of triangular complex spherical designs with small mesh ratio\thanks{\today.
\funding{The first and second authors acknowledge support
from the Australian Research Council under Discovery Project DP180100506.
The last author acknowledges support by Taiwan MOST Grant 107-2115-M-008-010-MY2.
This work was supported by the National Science Foundation under Grant No.~DMS-1439786 while the first, second and last authors were in residence at the Institute for Computational and Experimental Research in Mathematics in Providence, RI, during the Point Configurations in Geometry, Physics and Computer Science Program.
}}}

\author{
Yu Guang Wang\thanks{School of Mathematics and Statistics, University of New South Wales, Sydney, NSW, Australia; ICERM, Brown University, Providence, RI, USA 
  (\email{yuguang.wang@unsw.edu.au}).}
\and
Robert S. Womersley\thanks{School of Mathematics and Statistics, University of New South Wales, Sydney, NSW, Australia; ICERM, Brown University, Providence, RI, USA 
  (\email{R.Womersley@unsw.edu.au}).}
\and Hau-Tieng Wu\thanks{Department of Mathematics and Department of Statistical Science, Duke University, Durham, NC, USA; Mathematics Division, National Center for Theoretical Sciences, Taipei, Taiwan 
  (\email{hauwu@math.duke.edu}).}
\and Wei-Hsuan Yu\thanks{Department of Mathematics, National Central University, Taoyuan, Taiwan; ICERM, Brown University, Providence, RI, USA 
  (\email{u690604@gmail.com}).}}

\begin{document}

\maketitle

\begin{abstract}
    This paper provides triangular spherical designs for the complex unit sphere $\Omega^d$ by exploiting the natural correspondence between the complex unit sphere in $d$ dimensions and the real unit sphere in $2d-1$. The existence of triangular and square complex spherical $t$-designs with the optimal order number of points is established. A variational characterization of triangular complex designs is provided, with particular emphasis on numerical computation of efficient triangular complex designs with good geometric properties as measured by their mesh ratio.  We give numerical examples of triangular spherical $t$-designs on complex unit spheres of dimension $d=2$ to $6$.
\end{abstract}



\section{Introduction}
Uniform design that seeks design points uniformly distributed over a domain is an important statistical problem with many industrial and scientific applications \cite{fang2000uniform}. For example, chemists and chemical engineers have used the uniform design technique, uniformity of space-filling, to improve cost-efficiency, robustness and flexibility of experimental works.
For the regression problem, with a small number of sampling points, a significant amount of information can be obtained for exploring the relationship between the response and the contributing factors. Applications of uniform design can be found in the textile, pharmaceutical, and fermentation industries, among others. 
See \cite{liang2001uniform,fang1980uniform} for details and a review on the topic \cite{fang2000uniform} for the underlying theory and applications. 
Different measures of uniformity are appropriate for applications from statistics, numerical integration, minimum energy or geometry.
The construction of a uniform design depends critically on the domain and dimension. In this paper, we concentrate on the complex spheres $\csd\subset\cd$ for dimension $d\geq2$.
For $d = 1$, $\csd$ is the unit circle in the complex plane
and (a rotation of) the $N$th roots of unity are ideal.


In \cite{dette2019optimal,dryden2005statistical,koay2011simple,pronzato2012design,talke2016measures}, researchers study uniformly random designs on $2$ and $3$-dimensional spheres.
Real spherical $t$-designs on $\sph{d} \subset \Rd[d+1]$ for specific combinations of $t$ and $d$ have been constructed based on an algebraic construction (i.e. using a group orbit) see for example \cite{bannaibannai2009}.
However, there is only relatively limited work on higher dimensional real spheres \cite{bannaibannai2009} or complex spheres, particularly for deterministic designs.
Motivated by analyzing complicated time series and circumventing the randomness of a recently developed generalized multi-tapered time-frequency analysis tool~\cite{DaWaWu2016}, we consider the uniform design problem on complex spheres of various dimensions.
For the complex sphere $\csd$, triangular, square or rectangular complex designs can be defined. In this paper, we concentrate on triangular complex $t$-designs on $\csd$, which have a natural correspondence with real spherical $t$-designs on $\sph{2d-1}$, $d \geq 2$.
The aim is to provide sequences of efficient triangular complex $t$-designs (with the number of points roughly the dimension of the polynomial space divided by $2d-1$) for the complex unit sphere $\csd$ with good geometric properties, as measured by the mesh ratio of the point set.

While the optimal order for the number of points $N$ in a real spherical $t$-design has been well studied~\cite{BoRaVi2013}, the existence of the optimal order designs and the leading order term for complex spheres have not yet been considered. We generalize the existence of real spherical $t$-designs \cite{BoRaVi2013} to the complex case by applying the Brouwer degree theory and Marcinkiewicz-Zygmund type inequality on manifolds.
Moreover, we offer a deterministic triangular complex design scheme that is ``uniform'' and satisfies a good numerical integration property for the complex sphere of any dimension.
The scheme is a generalization of a recently developed optimization algorithm used in \cite{Womersley2018} for real spheres $\sph{d}$, to find efficient spherical $t$-designs for $\csd$, $d\geq 2$, with low mesh ratios.


The paper is organized as follows. In Section~\ref{Section:SD}, we review existing results for real spherical designs and the worst-case error for numerical integration. In Section~\ref{Section:ComplexSphereDesign}, we introduce the different types of spherical designs for complex spheres of any dimension. We provide a theoretical statement of the existence and the optimal order for triangular and square complex spherical $t$-designs. In Section~\ref{sec:uniformity}, we study the notions of separation and covering and provide a one-to-one mapping between $\sph{2d-1}$ and $\csd$, which preserves the geometric features of the mapped point set. In Section~\ref{sec:optsd}, we detail the numerical optimization scheme for quasi-uniform real spherical $t$-designs of any dimension, which one can use to construct the triangular complex spherical $t$-designs. In Section~\ref{sec:numer}, we show the numerically constructed real spherical $t$-designs on $\sph{d}$ for $d=3,5,7,9,11$, which can be used to construct triangular complex $t$-designs on $\csd$ for $d=2,3,4,5,6$. Then, we utilize the constructed triangular complex $t$-designs to illustrate the convergence for numerical integration on $\csd[2]$. In Section~\ref{sec:conclusion}, we give a brief discussion and conclusion. In Section~\ref{sec:proof}, we give the detailed proofs of the theoretical results in the previous sections.

\section{Real spherical designs}\label{Section:SD}
The spherical design scheme is originated from discrete geometry. The spherical design problem has been widely studied for various purposes, including numerical integration, interpolation, regularization model, kernel-based approximation, fast Fourier transforms, sparse signal recovery, and computation of random fields.
There are different notions of {\em uniformly distributed points} on the sphere; for example, the energy minimizing points \cite{SaTo1997} (also called the solution for Thomson problems) and equal area points \cite{Leopardi2006, RaSaZh1994} are both uniformly distributed, but in different senses. 
In this section, we summarize the notion of spherical design and uniformity on the $d$-dimensional real unit sphere $\sph{d} \subset \Rd[d+1]$.

A desirable property of the designed points on the sphere is that they reduce the number of function evaluations in numerical integration; that is, we can sample a limited number of points to obtain the integral of a polynomial with zero-loss.

\begin{definition}[\cite{DeGoSe1977}] 
Take $t\in \mathbb{N}$.
A finite subset $X$ of the $d$-dimensional real unit sphere $\sph{d}\subset\Rd[d+1]$
is called a  {\rsd} if for any real spherical polynomial
$p$ of degree at most $t$, we have
\begin{equation}
    \frac{1}{\sigma(\sph{d})} \int_{\sph{d}}p(\bx) \mathrm{d}\sigma(\bx) =
\frac{1}{|X|}\sum_{\bx \in X} p(\bx),
\end{equation}
where $\mathrm{d}\sigma=\mathrm{d}\sigma_d$ is the canonical Riemannian measure on $\sph{d}$.
\end{definition}

Note that if $p(\bx)$ is a harmonic polynomial, then its integral is zero. Therefore, an equivalent definition is that a set $X\subset \sph{d}$ is a  {\rsd} if
\begin{equation}
  \sum_{\bx \in X} p(\bx) = 0 
\end{equation}
for all $p(\bx) \in \text{Harm}_l(\Rd[d+1])$, where $1 \leq l \leq t$ and $\text{Harm}_t(\Rd[d+1])$ is the set of homogeneous harmonic polynomials of degree $t$ on $\Rd[d+1]$ \cite{DeGoSe1977}. 
If the number of points $N = |X|$ is sufficiently large, there always exists a real spherical $t$-design on $\sph{d}$, and the optimal order is $N=\bigo{}{t^d}$. Moreover, if a spherical $N$-point $t$-design on $\sph{d}$ exists, then there may exist multiple different (not related by rotation) such $t$-designs.

For the purpose of fast numerical integration, the  {\rsd s} perform better than other kinds of designs, such as randomly uniformly distributed points, equal area points \cite{RaSaZh1994,Leopardi2006}, Fekete points \cite{SlWo2004}, Coulomb energy \cite{SaTo1997} points,\footnote{Coulomb energy points minimize the functional $\sum_{i,j=1,\ldots,N,\,i\neq j}\frac{1}{|\bx_i-\bx_j|}$, which is a classic but open problem in electromagnetic dynamics.} log energy points \cite{SaTo1997},\footnote{Log energy points minimize the functional $\sum_{i,j=1,\ldots,N,\,i\neq j}\log\frac{1}{|\bx_i-\bx_j|}$.} generalized spiral points \cite{Bauer2000,RaSaZh1994} and distance points \cite{BrSaSlWo2014}.
To quantify the performance of numerical integration using a finite set of points $X_N=\{\bx_1,\dots,\bx_N\}$, where $N\geq2$, we may consider the \emph{$p$-th worst-case error} (wce), where $p>0$, defined on the Sobolev space $\sobH:=\sobH(\sph{d})$ with the smoothness $s>d/2$ \cite{BrSaSlWo2014,BrDiSaSlWaWo2014}:
\begin{equation}
    \wce{X_N,p;\sobH} := \sup_{\genfrac{}{}{0pt}{}{f\in \sobH}{\|f\|_{\sobH}\leq1}} \left|\frac{1}{N}\sum_{i=1}^Nf(\bx_i)-\int_{\sph{d}}f(\bx)\mathrm{d}\sigma(\bx)\right|^p\,.
\end{equation}
Moreover, given $p\geq1$, for a point set $X_N$, we define the {\em strength} $s^*=s^*_p$ to be the largest $s>0$ such that
\begin{equation}
(\wce{X_N,p;\sobH})^{1/p} = \mathcal{O}(N^{-s/d})\,.
\end{equation} 

By \cite[Theorem~7]{BrSaSlWo2014}, for the set $X^{\texttt{RU}}_N$ of $N$ i.i.d. randomly uniformly distributed points on the sphere (i.e. the points of $X^{\texttt{RU}}_N$ are drawn independently with the uniform distribution on $\sph{d}$), its worst-case error for $p = 2$ satisfies
\begin{equation}\label{WCErand}
    \sqrt{\expect{\wce{X^{\texttt{RU}}_N,2;\sobH}}} = \frac{c(s,d)}{\sqrt{N}},
\end{equation}
where the constant $c(s,d):=\sqrt{\sum_{\ell=1}^{\infty}(1+\lambda_\ell)^{-s}Z(d,\ell)}$ depends only on $s$ and $d$, where $\lambda_{\ell}=\ell(\ell+d-1)$, $\ell\geq0$, are the eigenvalues of the Laplace-Beltrami operator on $\sph{d}$ and $Z(d,\ell)$ is the dimension of the space of spherical harmonics of degree $\ell$ given by \eqref{eq:zdl} below. The constant $c(s,d)$ decays to zero as $s\to\infty$. On the other hand, the worst-case error for the  {\rsd} $X^{\texttt{t-SD}}_N$ has the upper bound
\begin{equation}\label{WCEStD}
    \wce{X^{\texttt{t-SD}}_N,1;\sobH} = \bigo{}{t^{-s}},
\end{equation}
for any $s>d/2$ \cite{Brandolini_etal2014}.
If we take $N=\mathcal{O}(t^d)$ (which is always achievable \cite{BoRaVi2013} and which we use in the optimization in Section~\ref{sec:optsd}) for the {\rsd}, we would obtain the following optimal order bound:
\begin{equation}\label{eq:StDwce}
    \wce{X^{\texttt{t-SD}}_N,1;\sobH} = \bigo{}{N^{-s/d}},
\end{equation}
for any $s>d/2$, where the big $\mathcal{O}$ is independent of $N$ but may depend on $s$ and $d$. We call the number $N=\mathcal{O}(t^d)$ the {\em optimal-order number} for the {\rsd s}. By \eqref{eq:StDwce} and \eqref{WCErand}, for fixed $s$, the worst-case error of  {\rsd s} has higher order convergence rate than the randomly uniformly distributed points.
Furthermore, for the randomly uniformly distributed points, by \eqref{WCErand} and the Cauchy-Schwartz inequality, the strength is $s^*_2=d/2$ on average. On the other hand, the   {\rsd} have the strength of $s^*_1=\infty$. Thus, compared to randomly uniformly distributed points, {\rsd s} perform better for the equal weighted numerical integration of smooth functions. 
The numerical evidence by Brauchart et al. \cite{BrSaSlWo2014} shows that on $\mathbb{S}^2$, the strengths for the Fekete, equal area, Coulomb energy, Log energy, generalized spiral, and distance are close to $s^*=1.5,2,2,3,3,4$ respectively. Thus, {\rsd}s are also superior to these points for numerical integration.

\section{Complex spherical designs}\label{Section:ComplexSphereDesign}
For a complex number $z$ in $\C$, let $|z|:=\sqrt{(\re\, z)^2+(\im\, z)^{2}}$ be the modulus of $z$, where $\re\, z$ and $\im\, z$ are the real and imaginary parts of $z$. Let $\imu:=\sqrt{-1}$ be the imaginary unit. 
Let $\innerh{\cu}{\cv}:=\sum_{j=1}^{d}u_j \conj{v_j}$ be the Hermitian inner product of $\cu,\cv\in \cd$, where $u_j$ and $v_j$ are the $j$th components of $\cu$ and $\cv$, and $\conj{v_j}$ is the complex conjugate of $v_j$.
Let $\|\cu\|:=\sqrt{\innerh{\cu}{\cu}}=\sqrt{\sum_{j=1}^{d}|u_j |^{2}}$ be the $\ell_2$ norm of $\cu\in\cd$.
For $d\ge1$, let 
\[
\csd := \{\cu\in\cd : \|\cu\|=1\}
\]
be the complex (unit) sphere of $\cd[d]$.
Let $\distc(\cu,\cv):=\arccos(\re\innerh{\cu}{\cv})$ be the geodesic distance between $\cu$ and $\cv$ in $\csd$.

Spherical designs can also be defined on the complex sphere $\csd$, see \cite{RoSu2014}. For $d\ge2$, let $\csdm:=\csdm[d]$ be the invariant Haar measure (and the normalized Lebesgue measure, which is unique) on $\csd$ satisfying $\int_{\csd}\intc{u}=1$. The volume of $\csd$ is
\begin{equation}\label{eq:surf.area.csd}
	\omega_d := \frac{2\pi^{d}}{\Gamma(d)},
\end{equation} 
where $\Gamma$ is the Gamma function,
see e.g. \cite{Rudin1980}.
Let $\Lpc{2}:=L_2(\csd,\csdm)$ be the Hilbert space of complex-valued square integrable functions on $\csd$ with respect to $\csdm$, and inner product
\begin{equation*}
	\innerlc{f,g} := \int_{\csd}f(\cu)\conj{g(\cu)}\intc{u}.
\end{equation*}
For $d\ge2$, let 
	$\Delta:=\sum_{j=1}^{d}\frac{\partial^{2}}{\partial u_j\partial \conj{u_j}}$
be the Laplacian operator on $\cd$. A function $f$ on $\cd$ is called harmonic if $\Delta f=0$.
Let $k,l \in\mathbb{N}_0$.
A polynomial $p$ of $\cu$ on $\cd$ is called homogeneous of degree $k$ if $p$ is homogeneous of degree $k$ in $\cu$, that is,
\begin{equation*}
	p(a \cu) = a^k p(\cu)\quad a\in \C,\; \cu\in \cd.
\end{equation*}
A polynomial $p$ of $\cu$ on $\cd$ is called homogeneous of degree $(k,l)$ if $p$ is homogeneous of degree $k$ in $\cu$ and homogeneous of degree $l$ in $\conj{\cu}$, that is,
\begin{equation*}
	p(a \cu) = a^k\conj{a}^l p(\cu)\quad a\in \C,\; \cu\in \cd.
\end{equation*}
The restriction to $\csd$ of a homogeneous harmonic polynomial of degree $(k,l)$ on $\cd$ is called a \emph{complex spherical polynomial of degree $(k,l)$} on $\csd$. 
Let $\mathcal{H}_{k,l}(\csd)$ be the space of all complex spherical polynomials of degree $(k,l)$ on $\csd$.
The spaces $\mathcal{H}_{k,l}(\csd)$ and $\mathcal{H}_{k',l'}(\csd)$ are orthogonal for $(k,l)\neq (k',l')$. The union (or the direct sum) of $\mathcal{H}_{k,l}(\csd)$ over all $k,l\in\Nz$ is dense in $\Lpc{2}$, that is,
\begin{equation*}
	\Lpc{2} = \bigoplus_{k,l=0}^{\infty}\mathcal{H}_{k,l}(\csd).
\end{equation*}
The restriction to $\csd$ of a homogeneous harmonic polynomial of degree $k$ on $\cd$ is called a \emph{complex spherical polynomial of degree $k$} on $\csd$. 
Let $\mathbb{H}_{t}(\csd)$ be the space of all complex spherical polynomials of degree at most $t\geq0$ on $\csd$. Then,
\begin{equation*}
	\mathbb{H}_{t}(\csd)  = \bigoplus_{k+l\leq t}\mathcal{H}_{k,l}(\csd).
\end{equation*}
The dimension of the space $\mathcal{H}_{k,l}(\csd)$ is
\begin{align}\label{eq:mdt}
	Z^{(d)}_{k,l} &:= \frac{(d-1)(k+l+d-1)\Gamma(d+k-1)\Gamma(d+l-1)}{(\Gamma(d))^{2}\Gamma(k+1)\Gamma(l+1)}\\
	&\asymp (k+l+1)(k+1)^{d-2}(l+1)^{d-2},\notag
\end{align}
where $a_{\ell}\asymp b_{\ell}$ for sequences $\{a_{\ell}\}_{\ell=0}^{\infty}$ and $\{b_{\ell}\}_{\ell=0}^{\infty}$ means that there exist constants $c,c'>0$ such that $c' a_{\ell}\leq b_{\ell}\leq c a_{\ell}$ for all $\ell\geq0$.
Then, the dimension of $\mathbb{H}_{t}(\csd)$ is
\begin{equation}
	M^{(d)}_{t}:=\sum_{k+l\leq t}Z^{(d)}_{k,l}
	= 
	\frac{(2d + 2t -1) \Gamma(2d + t - 1)}{\Gamma(2d) \Gamma(t+1)} = \frac{2}{\Gamma(2d)} t^{2d-1} +
	\bigo{}{t^{2d-2}}.
\end{equation}
For more detailed discussion on complex spherical polynomials, see \cite{CoKlSi2011}.
\begin{definition}
Let $d\ge2$ and $k,l\in\Nz$.
A finite point set $X \subset \csd$ satisfying
\begin{equation}\label{eq:csdnumerint}
 \int_{\csd}g(\bz) {\rm d}\mu(\bz) =
\frac{1}{|X|}\sum_{\bz \in X} g(\bz)\quad \forall g\in \mathcal{H}_{k',l'}(\csd),\; (k',l')\in \mathcal{T}
\end{equation}
with
\begin{enumerate}
\item $\mathcal{T}=\{(k',l')\in\Nz^2: k'+l'\le t\}$
	is a triangular complex spherical $t$-design 
\item $\mathcal{T}=\{0,1,\ldots,k\}\times\{0,1,\ldots,l\}$
	is a rectangular complex spherical $(k,l)$-design
\item $\mathcal{T}=\{0,1,\ldots,t\}\times\{0,1,\ldots,t\}$
	is a square complex spherical $t$-design
\end{enumerate}
for $\csd$.
\end{definition}
The rectangular and triangular complex spherical designs are related.
\begin{remark}\label{rem:tri.rect.csd}
	Fix $d\ge1$ and $k,l\in\Nz$.
\begin{enumerate}
	\item A triangular complex spherical $(k+l)$-design is a rectangular complex spherical $(k,l)$-design and $(l,k)$-design. Hence, a triangular $2t$-design is automatically a square $t$-design.
	\item A rectangular complex spherical $(k,l)$-design is a triangular complex spherical $\min\{k,l\}$-design. Hence, a square complex $t$-design is a triangular complex $t$-design.
	\item A square complex spherical $t$-design for $\csd$ is equivalent to the notion of a complex projective $t$-design \cite{RoSu2014, hoggar1989tight}. The projective $t$-designs are not only discussed in complex space but also in quaternions $\mathbb{H}$ and octonions $\mathbb{O}$ \cite{hoggar1982t}.
\end{enumerate} 
\end{remark}
$X\in \csd$ is called \emph{symmetric} (antipodal) if $\bx \in X$ $\iff$ $-\bx \in X$.
\begin{remark}
If $k'+l'$ is odd and the point set $X$ is symmetric, then \eqref{eq:csdnumerint} is automatically satisfied as both sides are zero.
\end{remark}

In this paper, we focus on triangular and square complex spherical $t$-designs on $\csd$.
The following theorem shows that there exists a triangular complex $t$-design and a square complex $t$-design with an optimal order number of points in $\Omega^d$ in general. We generalize the proof used in \cite{BoRaVi2013} to show the optimal order.
  
\begin{theorem}\label{thm:existence_complex_design}
Let $\Omega^d$ be the unit complex sphere of $\mathbb{C}^d$, $d\geq1$. Then, there exists a constant $C_d>0$ dependent only on $d$ such that for each $t\geq1$ and any $N\geq C_d t^{2d-1}$, there exists a triangular complex spherical $t$-design, with and $N$ nodes, and the order $t^{2d-1}$ is optimal.
\end{theorem}
\begin{remark}
As a triangular $2t$-design is a square $t$-design, Theorem~\ref{thm:existence_complex_design} shows that the optimal order of the square complex $t$-design is also $C'_d t^{2d-1}$.
Note that the constant $C'_d$ for the square case satisfies $C'_d\leq 2^{2d-1}C_d$.
\end{remark}

We postpone the proof to Section~\ref{sec:uniformproof}.

\section{Notions of uniformity}\label{sec:uniformity}
In addition to the efficient numerical integration, one natural question scientists might ask is how to design points over a sphere that are ``uniform'' \cite{saff1997distributing}?
While they are good candidates, in general, the  {\rsd s} might not be uniform on the sphere. Before we search for ``uniform  {\rsd s}'', we need to define the notion of uniformity on the sphere rigorously.
When $d=1$, the nodes of the {\rsd} are the roots of unity (probably with a global rotation). To be precise, if we want to distribute $N$ points uniformly on $\mathbb{S}^1$, we can choose $\{e^{i(2\pi k/N+\theta)}\}_{k=1}^{N}\subset \mathbb{S}^1$, where $\theta\in [0,2\pi /N)$. 
When $d>1$, the answer is tricky. Specifically, the notion of ``uniform'' for a deterministic point set is no longer unique, and there are several possible definitions, including uniformity, quasi-uniformity and hyper-uniformity. Uniformity can be described by Weyl's formula: a set $X_N=\{\bx_1,\dots,\bx_N\}\subset\sph{d}$ is called \emph{uniformly distributed} if for any $\bx\in \sph{d}$ and $r\in (0,\pi/2)$,
\begin{equation}
\lim_{N\to\infty}\frac{1}{N}\sum_{i=1}^N\mathbf{1}_{C(\bx,r)}(\bx_i)=\frac{\sigma(C(\bx,r))}{\sigma(\sph{d})},
\end{equation}
where $C(\bx,r):=\{\by\in\sph{d}: \arccos(\bx\cdot\by)\leq r\}$ is the spherical cap of radius $r$ centered at $\bx\in \sph{d}$. This intuitively says that the number of points in any spherical cap with the same radius is asymptotically the same. Equivalently, the uniformity can be characterized by {\em spherical cap discrepancy} \cite{BrDi2012,BrSaSlWo2014,KuNi1974,Stolarsky1973}.
The uniformity was then generalized to hyper-uniformity, which allows characterizing the distribution of \emph{determinantal point processes}, as motivated by the statistical physics applications; we refer the readers with the interest to  \cite{BrGrPeKu2018det, BrGrKu2018prob, Torquato2016, ToSt2003} for details.

{\em Quasi-uniformity} \cite{BrDiSaSlWaWo2014,MhNaPrWa2010} is a measure of the geometric quality of a sequence of point sets based both on the separation and the mesh norm of the point set.

\begin{definition}\label{defn:meratr}
Let $X_{N}:=\{\bx_{1},\dots,\bx_N\}\subset \sph{d}$, where $d\ge2$ and $N\geq 2$. The separation distance of $X_{N}$ is defined as the minimal distance between the points of $X_{N}$,
\begin{equation}
    \sepr_{X_{N}}:=\min_{\stackrel{1\le i,j\le N}{i\neq j}}\distr(\bx_{i},\bx_{j})\,,
\end{equation}
where $\distr(\cdot,\cdot)$ is the geodesic distance on $\sph{d}$.
The covering distance (which is also called the mesh norm or fill distance) of $X_{N}$ is 
\begin{equation}
    \coverr_{X_{N}}:=\max_{\by\in \sph{d}}\min_{1\le i\le N}\distr(\by,\bx_{i}).
\end{equation}
Finally, the mesh ratio of $X_{N}$ is 
\begin{equation}\label{eq:meshratio}
    \mratr_{X_{N}}:=\mratr(X_{N}):=\frac{2\coverr_{X_{N}}}{\sepr_{X_{N}}} \geq 1.
\end{equation}
\end{definition}

The covering distance $\coverr_{X_{N}}$ is the minimal radius, for which identical spherical caps with the points of $X_{N}$ as centers cover the whole sphere, while $\sepr_{X_{N}} / 2$ is the radius for packing the sphere with non-overlapping spherical caps with centers in $X_N$.

\begin{definition}[\cite{Brauchart_etal2018random,BrSaSlWo2014,Womersley2018}]
A sequence of point sets $\{X_{N}\}$ on $\sph{d}$ is $C$-quasi-uniform if $\rho_{X_N}\leq C$ for all point sets in the sequence, where $C$ is independent of $N$ (but may depend on $d$).
\end{definition} 

Thus the closer the mesh ratio is to $1$, the more ``uniform'' the point set is. 
While there are several notions of uniformity, quasi-uniformity is computationally friendly, as it relates the largest hole in the design
to the separation of the points of the design.

By \cite[Corollary~1.7]{BrDiSaSlWaWo2014}, the covering radius of a  {\rsd} on $\sph{d}$ with the optimal-order number of nodes, $N=\bigo{}{t^d}$, has the optimal-order bound
\begin{equation}\label{eq:stdcovering}
    \coverr_{X^{\texttt{t-SD}}_N}\leq c/N^{1/d},
\end{equation}
where the constant $c$ depends only on the dimension $d$. 
By 
\cite{Brauchart_etal2018random,ReSa2016}, for $N$ randomly uniformly distributed points on $\sph{d}$, asymptotically when $N\to \infty$ we have
\begin{equation}
    \expect{\coverr_{X^{\texttt{RU}}_N}}\asymp \left(\frac{\log N}{N}\right)^{1/d},
\end{equation}
where the implied constant depends only on $d$.
%
On the other hand, the separation of the real spherical $t$-design can be arbitrarily small, since any rotation of a real spherical $t$-design is again a real spherical $t$-design, and the union of two  {\rsd}'s is a real spherical $t$-design. To have the mesh ratio of the spherical design bounded, one needs to select from different  {\rsd}'s to ensure the points are well separated. It is one of the key steps in the proposed spherical design scheme introduced in Section~\ref{sec:optsd}.

For the complex spheres, in a similar way to the real case, the \emph{separation} and \emph{covering} radii of a point set ${X}_{N}:=\{\bz_1,\dots,\bz_N\}\subset \csd$ are defined as
\begin{equation}
    \sepc_{X_{N}} := \min_{\stackrel{1\le i,j\le N}{i\neq j}}\distc(\bz_{i},\bz_{j}),\quad
    \coverc_{X_{N}} := \max_{\bz\in \cd}\min_{1\le i\le N}\distc(\bz,\bz_{i})\,,
\end{equation}
and the \emph{mesh ratio} is defined as
\begin{equation}
    \mratc_{X_N} := \frac{2\coverc_{X_{N}}}{\sepc_{X_{N}}}\geq 1.
\end{equation}
\begin{definition}
A sequence of point sets $\{X_{N}\}\subset\csd$ is $C$-quasi-uniform if there exists a constant $C$ such that $\mratc_{X_{N}}\le C$ for all point sets $X_N$ in the sequence. 
\end{definition} 


To further study the relationship between the  {\rsd} and \cpsd, we consider the following mapping $\phi^{\texttt{R}\to \texttt{C}}:\sph{2d-1}\mapsto\Omega^{d}$ for $d\geq2$. For $j=1,\dots,d$, the real and complex parts of the $j$th component of $\phi^{\texttt{R}\to \texttt{C}}(\bx)$ for $\bx\in\sph{2d-1}$ are $x_{2j-1}$ and $x_{2j}$:
\begin{equation}\label{eq:map.sph.csd}
    (\phi^{\texttt{R}\to \texttt{C}}(\bx))_{j} = x_{2j-1} +  x_{2j}\imu
\end{equation}
where $\imu^2 = -1$.
The map $\phi^{\texttt{R}\to \texttt{C}}$ is well-defined as $|\phi^{\texttt{R}\to \texttt{C}}(\bx)|^{2}=|\bx|^{2}=1$.
For a point set $X_N:=\{\bx_1,\dots,\bx_N\}\subset\sph{2d-1}$, we let $\phi^{\texttt{R}\to \texttt{C}}(X_N):=\bigl(\phi^{\texttt{R}\to \texttt{C}}(\bx_1),\dots,\phi^{\texttt{R}\to \texttt{C}}(\bx_N)\bigr)$.
Note that the set $X_N$ is symmetric if and only if $\phi^{\texttt{R}\to \texttt{C}}(X_N)$ is symmetric.

    To study the quasi-uniformity for complex spheres, we have the following proposition, which shows that $\phi^{\texttt{R}\to \texttt{C}}$ preserves the separation, covering and mesh ratio. It implies that if $X_N\subset \sph{2d-1}$ is quasi-uniform, so is $\phi^{\texttt{R}\to \texttt{C}}(X_N)\subset\csd$.

\begin{proposition}\label{Proposition: csd rsd relationship}
    Let $\phi^{\texttt{R}\to \texttt{C}}$ be the mapping given by \eqref{eq:map.sph.csd}. Then
    \begin{equation}\label{eq:sep.cover.meshratio.sph.cd}
        \sepr_{X_N} = \sepc_{\phi^{\texttt{R}\to \texttt{C}}(X_N)},\quad \coverr_{X_N} = \coverc_{\phi^{\texttt{R}\to \texttt{C}}(X_N)},\quad \mratr_{X_N} = \mratc_{\phi^{\texttt{R}\to \texttt{C}}(X_N)}.
    \end{equation}
\end{proposition}
We give the proof in Section~\ref{sec:uniformproof}.

We mention that the mapping $\phi^{\texttt{R}\to \texttt{C}}$ in \eqref{eq:map.sph.csd} preserves the distance.
By evaluating the $\ell_2$ norm of a complex vector, we have 
\begin{align*}
    \|\phi^{\texttt{R}\to \texttt{C}}(\bx)-\phi^{\texttt{R}\to \texttt{C}}(\by)\|^{2} 
    &= 2-2\re\bigl(\phi^{\texttt{R}\to \texttt{C}}(\bx)\cdot\conj{\phi^{\texttt{R}\to \texttt{C}}(\by)}\bigr)\\ 
     &= 2 -2 \cos\bigl(\distc(\phi^{\texttt{R}\to \texttt{C}}(\bx),\phi^{\texttt{R}\to \texttt{C}}(\by))\bigr)\,,
\end{align*}
where the $\distc(\bz,\bz')$ is the geodesic distance of $\bz,\bz'\in \csd$.
By \eqref{eq:map.sph.csd},
\begin{align*}
    \|\phi^{\texttt{R}\to \texttt{C}}(\bx)-\phi^{\texttt{R}\to \texttt{C}}(\by)\|^{2} = \|\bx-\by\|^{2} = 2-2\:\bx\cdot \by = 2-2\cos\bigl(\distr(\bx,\by)\bigr)\,,
\end{align*}
where the $\distr(\bx,\by)$ is the geodesic distance of $\bx,\by\in\sph{2d-1}$.
Then, for $d\ge2$, we have
\begin{align}
    &\re\bigl(\phi^{\texttt{R}\to \texttt{C}}(\bx)\cdot\conj{\phi^{\texttt{R}\to \texttt{C}}(\by)}\bigr) = \bx\cdot \by\notag\\
    & \distc(\phi^{\texttt{R}\to \texttt{C}}(\bx),\phi^{\texttt{R}\to \texttt{C}}(\by)) = \distr(\bx,\by),\label{eq:distc.equiv.distr}
\end{align}
where $\bx,\by\in\sph{2d-1}$, which establishes the claim.

\section{Proposed Algorithm for Spherical Designs}\label{sec:optsd}
We now introduce the proposed deterministic spherical designs that are {\rsd}, and hence {\tcsd}, satisfying the quasi-uniformity. We exploit the high-dimensional numerical optimization \cite{AnChSlWo2010,ChFrLa2011,ChWo2006,ChWo2018,Womersley2018} to numerically implement this real spherical design.

\subsection{Variational characteristics of  {\rsd}}\label{sec:varchar}
We now give an equivalent characterization of a  {\rsd} in terms of {\em zonal kernel functions}. 
\begin{definition}
Take a Hilbert space $H$ with the inner product $x\cdot y$ for $x,y\in H$. A \emph{zonal (kernel) function} $K(x,y)$ is a mapping from $H\times H$ to $\R$ that depends only on the inner product of $x\cdot y$.
\end{definition}

For $\alpha,\beta>0$ and $\ell=0,1,\dots$, let $\Jcb{\ell}(u)$, $u\in[-1,1]$ be the Jacobi polynomial of degree $\ell$ with the weight $w_{\alpha,\beta}(u):=(1-u)^{\alpha}(1+u)^{\beta}$. The Jacobi polynomials $\Jcb{\ell}$ form an orthogonal system for the weighted $L_2$ space $L_2(w_{\alpha,\beta})$ on $[-1,1]$. See \cite[Chapter~4]{Szego1975} for details.
For $\ell\ge0$, the \emph{(normalized) Legendre polynomial} (or normalized ultraspherical or Gegenbauer polynomials) of degree $\ell$ for $d\ge2$ is the normalized Jacobi polynomial of degree $\ell$ with the weight $(1-u^{2})^{\frac{d-2}{2}}$:
\begin{equation}
	\Legen{\ell}(u) =\frac{ \Jcb[\frac{d-2}{2},\frac{d-2}{2}]{\ell}(u)}{\Jcb[\frac{d-2}{2},\frac{d-2}{2}]{\ell}(1)},\quad u\in[-1,1].
\end{equation}
The $\Legen{\ell}(u)$, $\ell\ge0$, form an orthogonal system with weight $(1-u^{2})^{\frac{d-2}{2}}$ on $[-1,1]$.
By \cite[Theorem~7.32.1]{Szego1975},
\begin{equation}
	|\Legen{\ell}(u)|\le 1,\quad u\in[-1,1].
\end{equation}
For $\ell\ge0$, the $\Legen{\ell}(\bx\cdot \by)$, $\bx,\by\in\sph{d}$ is a zonal function and satisfies the \emph{addition theorem}
\begin{equation}\label{eq:addthm}
	\sum_{m=1}^{Z(d,\ell)}\shY(\bx)\shY(\by) = \frac{Z(d,\ell)}{\sigma(\sph{d})}\Legen{\ell}(\bx\cdot \by),
\end{equation}
where $Z(d,\ell)$ is the dimension of the space of spherical harmonics of degree $\ell$ given by
\begin{equation}\label{eq:zdl}
	Z(d,\ell)=(2\ell+d-1)\frac{\Gamma(\ell+d-1)}{\Gamma(d)\Gamma(\ell+1)}\,.
\end{equation}
Note that asymptotically when $\ell\to \infty$, $Z(d,\ell)\asymp (\ell+1)^{d-1}$.

The following shows an equivalent characterization of {\rsd} by a radial basis function on the sphere. Let the radial basis function $\spdka(u)$ be a real-valued polynomial on $[-1,1]$ of degree $t$ with the expansion in terms of Legendre polynomials
\begin{equation}\label{eq:psit.expan.legen}
	\spdka(u) = \sum_{\ell=0}^{t}a_{\ell}\Legen{\ell}(u),\quad u\in [-1,1],
\end{equation}
where the coefficients $a_{\ell} := \int_{-1}^{1}\spdka(u)\Legen{\ell}(u)(1-u^{2})^{\frac{d-2}{2}}\IntD{u}$, $\ell =0,\dots,t$.
In particular, $a_{0} = \int_{-1}^{1}\spdka(u)(1-u^{2})^{\frac{d-2}{2}}\IntD{u}$.
We let $\spdk$ be the $\spdka$ subtracting the first term (i.e. the constant term) in the Legendre expansion:
\begin{equation}\label{eq:spdk}
	\spdk(u) := \spdka(u) - a_0,\quad u\in[-1,1].
\end{equation}

\begin{proposition}[\cite{SlWo2009,Womersley2018}]\label{prop:variat.char.sph.des}
	For $t\in \mathbb{N}$, let $\spdk$ be a real-valued polynomial on $[-1,1]$ of degree $t$ with any sequence of coefficients satisfying $a_{\ell}>0$, $\ell=1,\dots,t$ in the expansion \eqref{eq:psit.expan.legen}. Then, the point set $X_{N}:=\{\bx_1,\dots,\bx_N\}\subset \sph{d}$ is a  {\rsd} if and only if
	\begin{equation}\label{eq:phit.re.sph.des}
		\sum_{i=1}^{N}\sum_{j=1}^{N}\spdk(\bx_i\cdot\bx_j)=0.
	\end{equation}
\end{proposition}

\begin{example}\label{eg:zonalpolyreal}
An example of a zonal polynomial $\spdk$, which has a simple formula and strictly positive coefficients $a_{\ell}$, $\ell=1,\dots,t$, in the Legendre expansion is
(\cite{SlWo2009} for $\sph{2}$ and from \cite{Womersley2018} for $\sph{d}, d\geq3$):
	\begin{equation}\label{eq:eg3.psi}
		\spdk(u) = P_{t}^{(\frac{d}{2},\frac{d-2}{2})}(u)-1,\quad u\in [-1,1],
	\end{equation}
	where $a_{\ell}= Z(d,\ell)$ for $\ell=1,\dots,t$. The simple formula for $\spdk$ in \eqref{eq:eg3.psi} provides a well-scaled objective function which is easy to evaluate and differentiate for a numerical optimization procedure to search for {\rsd} with a large number of nodes on $\sph{d}$.  
\end{example}

\subsection{Numerical computation of triangular complex spherical designs}
From the above discussion, we can utilise real spherical $t$-designs on $\sph{2d-1}$ to find triangular complex $t$-designs on $\csd$ for $d\geq 2$.
By Proposition \ref{prop:variat.char.sph.des}, for a {\rsd} on $\sph{2d-1}$, ideally, one would solve the minimization problem
\begin{equation}\label{eq:optproblsd}
\begin{array}{cl}
    \displaystyle\min_{X_N\subset\sph{2d-1}} & \rho(X_N)\\ 
    \mbox{\small Subject to} & V_{t,N,\spdk}(X_N) = 0
\end{array}
\end{equation}
where $X_N:=\{\bx_1,\dots,\bx_N\}$ is a point set on $\sph{2d-1}$,
\begin{equation} \label{eq:V}
    V_{t,N,\spdk}(X_N):= \frac{1}{N^2}\sum_{i=1}^{N}\sum_{j=1}^{N}\spdk(\bx_i\cdot\bx_j)\,,
\end{equation}
and $\rho(X_N)$ is the mesh ratio of the point set $X_N$ given by \eqref{eq:meshratio}.
Note that there are in total $N(2d-1)$ variables in \eqref{eq:optproblsd}. In fact, except for a measure zero subset, $\bx:=(x_1,\dots,x_{2d})\in\sph{2d-1}$ can be parametrized by
\begin{align}
    x_1 &= \cos(\phi_1),\nonumber\\
    x_i &= \Pi_{j=1}^i\sin(\phi_j)\cos(\phi_i), \quad i=2,\dots,2d-1,\label{Parametrization Sd}\\
    x_{2d} &= \Pi_{j=1}^{2d-1}\sin(\phi_j)\nonumber
\end{align}
where $\phi_1\in[0,\pi]$ and $\phi_2,\dots,\phi_{2d-1}\in[0,2\pi]$.
As both $V_{t,N,\spdk}(X_N)$ and $\rho(X_N)$ are invariant to a rotation of the point set $X_N$, the point set can be normalized so for $j = 1,\ldots,2d-1$, the point $\bx_j$ has components $j+1,\ldots,2d$ zero.
Taking the rotation of the point set on $\mathbb{S}^{2d-1}$ into account
reduces the number of variables to
$(2d-1)(N-d)$ (assuming $N \geq 2d$) 
where we have taken account of the dimension $d(2d-1)$ of rotation matrix on $\sph{2d-1}$. Clearly, this is a large-scale non-linear optimization problem when $N$ is large.
In the following, we focus on some key issues of the numerical solution of \eqref{eq:optproblsd}.
One complicating factor is that $\rho(X_N)$ contains max-min factors producing a non-smooth function.

The number $N$ of nodes of the {\rsd} on $\sph{2d-1}$ determines the difficulty of the optimization problem: the fewer, the easier the optimization is. However, to satisfy the definition of {\rsd}, $N$ cannot be as small as we wish. A lower bound $N^*_t$ for $N$ was established in \cite{DeGoSe1977} using the property of spherical harmonics: 
\begin{equation}
    N^*_t :=\left\{\begin{array}{ll}
    \displaystyle 2\genfrac(){0pt}{}{2d+k-1}{2d-1}, & \mbox{for $t = 2k+1$ is odd},\\[4mm]
    \displaystyle  \genfrac(){0pt}{}{2d+k-1}{2d-1} + \genfrac(){0pt}{}{2d+k-2}{2d-1}, & \mbox{for $t = 2k$ is even}.
    \end{array}\right.
\end{equation}
This means that a {\rsd} on $\sph{2d-1}$ must contain at least $N^*_t$ points where $N^*_t = \frac{2^{2-2d}}{\Gamma(2d)} t^{2d-1} + \bigo{}{t^{2d-2}}$. 

In general, the minimization problem \eqref{eq:optproblsd} for a {\rsd} with $N^*_t$ nodes (which is called \emph{tight {\rsd}}) does not always have a solution as on $\sph{2d-1}$, $d\geq2$,{\rsd} with $N^*_t$ nodes exist for a few small values of $t$ \cite{bannaibannai2009}. Thus, one has to use more than $N^*_t$ points to obtain a feasible problem \eqref{eq:optproblsd}.
By \cite{DeGoSe1977,SlWo2009}, the condition \eqref{eq:phit.re.sph.des} with positive coefficients $a_{\ell}$, $\ell=1,\dots,t$, is equivalent to the equations 
\begin{equation}\label{eq:optsd2}
    r_{\ell,m}(X_N) := \sum_{j=1}^N \shY(\bx_j)=0,
    \qquad m=1,\dots,Z(2d-1,\ell), 
    \quad \ell=1,\dots,t,
\end{equation}
where $Z(2d-1,t)$ is given by \eqref{eq:zdl}. Then, there are $\sum_{\ell=1}^{t}Z(2d-1,\ell)=Z(2d,t)-1$ conditions. To solve \eqref{eq:optproblsd}, the number of free parameters $(2d-1)(N-d)$ generically should be at least the number of conditions. Thus, we can use $\widehat{N}_t$ points for the spherical $t$-design, where
\begin{align} 
    \widehat{N}_t & :=
    \left\lceil\frac{1}{2d-1}\left((2t+2d-1)
    \frac{\Gamma(t+2d-1)}{\Gamma(2d) \Gamma(t+1)}-1\right)+d\right\rceil \label{eq:Nhat1}\\
    & = \frac{2}{(2d-1)\Gamma(2d)}
    \left(
    t^{2d-1} + (2d^2-2d+\tfrac{1}{2}) t^{2d-2} + \bigo{}{t^{2d-3}} \right). \label{eq:Nhat2}
\end{align}
In the numerical optimization problem \eqref{eq:optproblsd}, we can use the spherical parametrization of $\widehat{N}_t$ points on $\sph{2d-1}$. To find a feasible point satisfying the spherical design constraint $V_{t,N,\spdk}(X_N) = 0$ to high accuracy per se is a challenging problem.

We may further take the symmetry of $\sph{2d-1}$ into account.
As any symmetric point set $X_N$ satisfies $\sum_{i=1}^N\shY(\bx_i)=0$ for all odd degree $\ell\geq 1$ and $m = 1,\ldots,Z(2d-1,\ell)$.
The number of constraints from the even degree polynomials are
$\sum_{\ell=1}^{(t-1)/2}Z(2d-1,2\ell) = \frac{\Gamma(t+2d-1)}{\Gamma(2d)\Gamma(t)}-1$, where $t$ is odd.
We can then use $\overline{N}_t$ points, where
\begin{align} 
    \overline{N}_t & := 2\left\lceil\frac{1}{2d-1}
    \left(\frac{\Gamma(t+2d-1)}{\Gamma(2d)\Gamma(t)}-1   \right)+d\right\rceil \label{eq:Nbar1}\\
    & = \frac{2}{(2d-1)\Gamma(2d)}
    \left(
    t^{2d-1} + (2d^2-3d+1) t^{2d-2} + \bigo{}{t^{2d-3}} \right). \label{eq:Nbar2}
\end{align}
The numerical optimization uses the spherical parametrization of $\overline{N}_t/2$ points, with the symmetry constraints 
$X_{\overline{N}_t} = [X_{\overline{N}_t/2} \ {-\!X_{\overline{N}_t/2}}]$
used in the evaluation of the spherical design constraints and objective in (\ref{eq:optproblsd}).
See \cite{Womersley2018} for details.

Note that the optimal order leading terms $t^{2d-1}$ in \eqref{eq:Nhat2} and \eqref{eq:Nbar2} are the same with coefficient $2/((2d-1)\Gamma(2d))<1$.  The imposition of symmetry only makes a difference in the coefficient of $t^{2d-2}$.
However for fixed $t$, the optimization problem for a symmetric point set has roughly half the number of variables and half the number of constraints.


By \cite{BoRaVi2014} and \eqref{eq:stdcovering}, a sequence of {\rsd}s for $\sph{2d-1}$ with the optimal order for the number of points and good mesh ratio always exists. 

The optimization can proceed as follows. Fix $d,t\in \mathbb{N}$, and select $N\in \mathbb{N}$ (for example $N = \widehat{N}_t$).
The algorithm is composed of two steps, which may need to be repeated using different starting points:
\begin{itemize}
\item[] (Step 1)~~Find a feasible point for \eqref{eq:optproblsd} by mininizing $V_{t,N,\spdk}(X_N)$ with an appropriate initial point set. 
    \item[] (Step 2)~~If the feasible point is not isolated use it as a starting point to solve \eqref{eq:optproblsd}.
\end{itemize}
Finding an initial point set with low mesh ratio for minimizing \eqref{eq:optproblsd} is critical. One option is to use an initial set of points with good geometry, for example, equal area points for $\sph{2d-1}$ and $d\geq2$.
The large-scale minimization problem \eqref{eq:optproblsd} is difficult as there typically are many local minima with $V_{t,N,\spdk}(X_N) > 0$, so $X_N$ is not a spherical design.

Here, we find a feasible point for \eqref{eq:optproblsd} such that $ V_{t,N,\spdk}(X_N) = 0$ by using the zonal polynomial to evaluate $V_{t,N,\spdk}(X_N)$ and its gradients with respect to the spherical parametrization of $X_N$.
Note that there is a concise formula for zonal polynomials, see Example~\ref{eg:zonalpolyreal}.
As in \cite{Womersley2018} the variational form can be expressed as a strictly positive sum of squares $V_{t,N,\spdk}(X_N) = \frac{1}{N^2}\br^T \bD \br$  where $\bD$ is a diagonal matrix with positive elements on the diagonal.\footnote{On $\sph{2}$, \eqref{eq:optsd2} provides a better form of the variational function \eqref{eq:optproblsd} as the sum of square structure can be exploited. However, spherical harmonics corresponding to higher dimensions are difficult to compute.} 
If $t$, $d$ are such that the terms inside the ceiling functions in \eqref{eq:Nhat1} or \eqref{eq:Nbar1} are integers, so the 
number of points is such that the number of parameters exactly matches the number of constraints, and the Jacobian of $\br$ is full rank, then the Hessian of $V_{t,N,\spdk}(X_N)$ is positive definite and a point with $V_{t,N,\spdk}(X_N) = 0$ (equivalently $\br = 0$) is isolated. Typically, there are still many isolated feasible points with different mesh ratios $\rho(X_N)$. If the term inside the ceiling functions is not an integer, then there are some degrees of freedom which we can use to minimize the mesh ratio.

\section{Numerical Examples}\label{sec:numer}
Figures~\ref{fig:ssd1} and \ref{fig:ssd2} show the features of a symmetric real spherical $23$-design on $\sph{3}$ (or equivalently $\Omega^2$) with $1,184$ nodes, which is obtained by solving \eqref{eq:optproblsd} with the zonal polynomial in \eqref{eq:eg3.psi}. The stereographic projection of the point set from $\sph{3}$ to $\Rd[3]$ using the pole $\mathbf{e}_1=(1,0,0,0)^T$ is shown in the Figure~\ref{fig:ssd1}. The left panel in  Figure~\ref{fig:ssd2} shows the $700,336$ ordered inner products, the largest of which gives the cosine of the separation distance. The right panel shows the $7,192$ ordered circumradii of the facets of the design, the largest of which gives the covering radius.
\begin{figure}[th]
  \centering
  \includegraphics[width=0.55\textwidth]{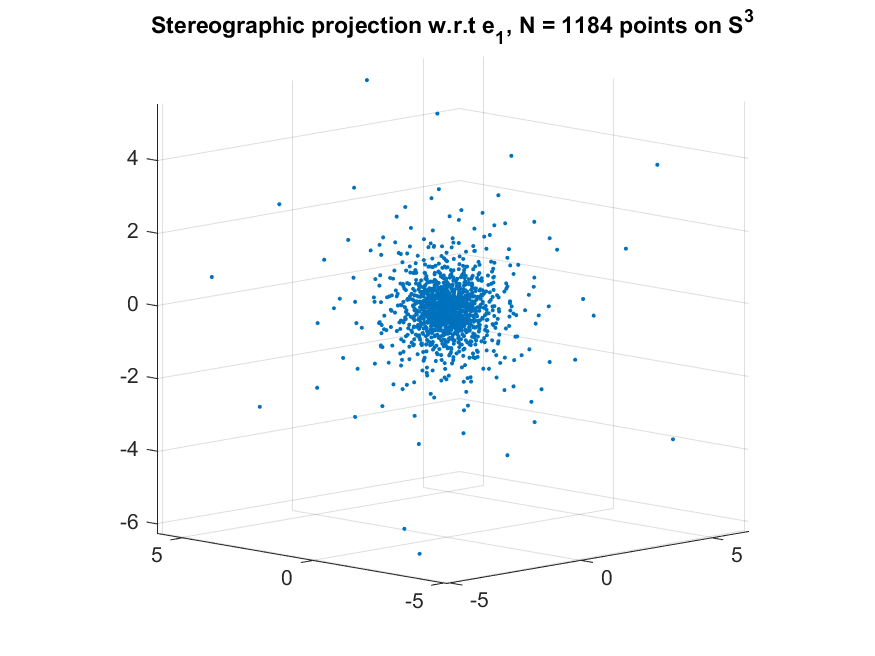}\hspace{2mm}
  \caption{Symmetric real spherical $21$-design on $\sph{3}$ with $1184$ nodes used for $\csd[2]$: Stereographic projection of the nodes onto $\Rd[3]$ with $(1,0,0,0)$ as the pole.}
\label{fig:ssd1}
\end{figure}
  
\begin{figure}[th]
  \begin{minipage}{\textwidth}
  \centering
  \includegraphics[width=0.47\textwidth]{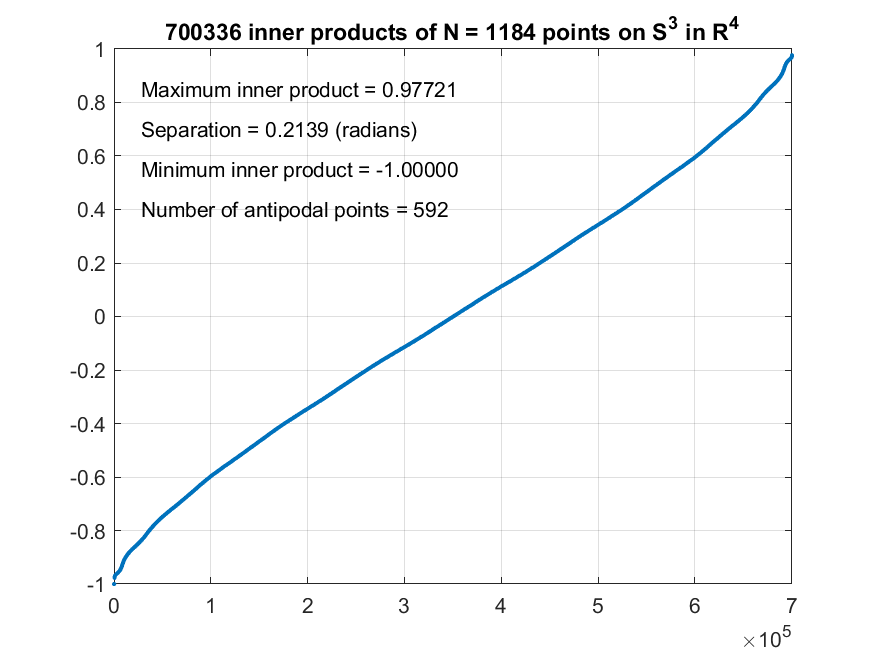}\hspace{2mm}
  \includegraphics[width=0.47\textwidth]{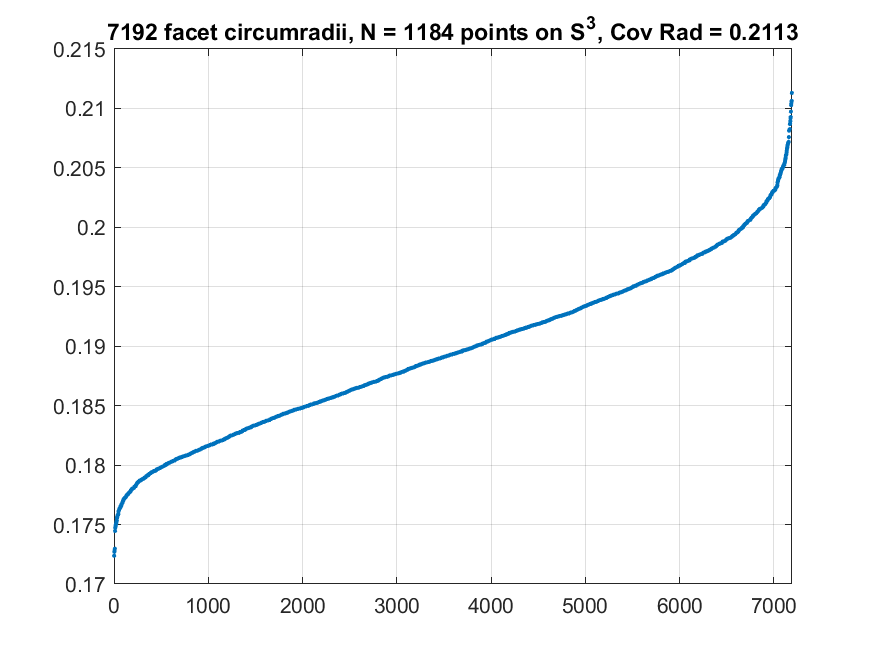}
  \end{minipage}
\begin{minipage}{\textwidth}
\caption{Symmetric real spherical $23$-design on $\sph{3}$ used for $\csd[2]$, with $1,184$ nodes. Left: $700,336$ sorted inner products between nodes with the largest giving the cosine of the separation distance of the design is $0.2139$. Right: Sorted $7,192$ facet circumradii, largest giving the covering radius of the design is $0.2113$.}
\label{fig:ssd2}
\end{minipage}
\end{figure}
Tables~\ref{tab:csd2}--\ref{tab:csd6} show some real spherical $t$-designs with good mesh ratios on $\sph{d}$, $d=3, 5, 7, 9, 11$ from Womersley\footnote{\url{https://web.maths.unsw.edu.au/~rsw/Sphere/EffSphDes/S3SD.html}} \footnote{\url{https://web.maths.unsw.edu.au/~rsw/Sphere/EffSphDes/SdSD.html}} \cite{Womersley2018}. They are used to construct triangular complex spherical $t$-designs on $\csd$ for $d=2,3,4,5,6$.

Tight triangular complex $t$-designs on $\csd$ exist only for certain small values of $t$, see \cite{bannaibannai2009} for the real case. In particular any pair of antipodal points is a tight triangular complex $1$-design for $\csd$ with separation angle $\pi$, covering radius $\pi/2$ and mesh ratio $1$.
A regular simplex of $N = 2d + 1$ points on $\csd$ is a tight triangular complex $2$-design on $\csd$ with mesh ratio $2 \cos^{-1}(1/2d) / \cos^{-1}(-1/2d)$.
Also, the cross-polytope with the $N = 4d$ points $\pm e_j, \pm i e_j, j = 1,\ldots,d$ on $\Omega^d$ is a tight triangular complex $3$-design on $\csd$. As this is an antipodal point set, it is only necessary to establish that complex polynomials in $\mathcal{H}_{k,l}$ for $k + l = 2$ are integrated exactly. For this point set the separation distance is $\pi/2$ and the covering radius is $\cos^{-1}(1/\sqrt{2d})$,
giving a mesh ratio of $(4/\pi) \cos^{-1}(1/\sqrt{2d})$.

As an example of  numerical integration by {\tcsd} on $\csd[2]$ consider the integrand $f(\bx)=1/{|\bx- \bx_0|^2}$, $\bx\in\csd[2]$.
For $|\bx_0| > 1$ the integrand is in $C^{\infty}(\csd[2])$. Here $\bx_0 = (1+\imu, 1+\imu)$, so $|\bx_0|^2 = 4$.
The exact value of the integral
\begin{equation*}
    \int_{\csd[2]}\frac{1}{|\bx-\bx_0|^2}\intc{x} = \frac{1}{|\bx_0|^2}.
\end{equation*}
Figure~\ref{fig:intf} shows the absolute value of the integration error 
\begin{equation*}
    \left|\frac{1}{N}\sum_{i=1}^N f(\bx_i) - \int_{\csd[2]}f(\bx)\intc{x}\right| = \left|\frac{1}{N}\sum_{i=1}^N f(\bx_i) - \frac{1}{|\bx_0|^2}\right|
\end{equation*}
for {\tcsd} with $t$ up to $31$, where the {\tcsd} is constructed from the real symmetric spherical $t$-designs in Table~\ref{tab:csd2}. It is clear the error for numerical integration of $f(\bx)$ by {\tcsd} is around $10^{-9}\sim10^{-12}$ for $t>21$. It illustrates that {\tcsd} in Table~\ref{tab:csd2} can provide precise estimate for integral over the complex sphere.

\begin{figure}
    \centering
    \includegraphics[width=0.48\textwidth]{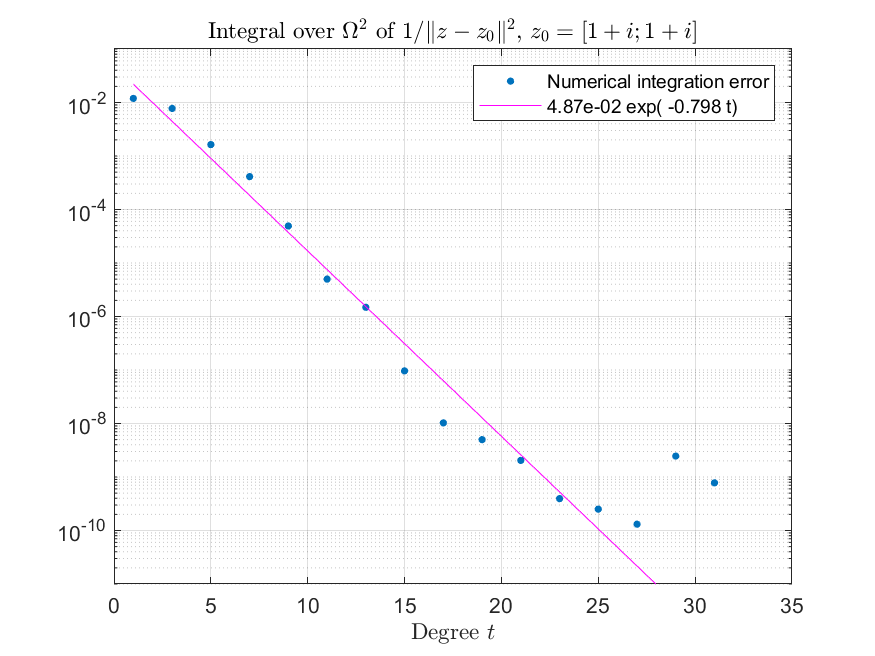}
    \includegraphics[width=0.48\textwidth]{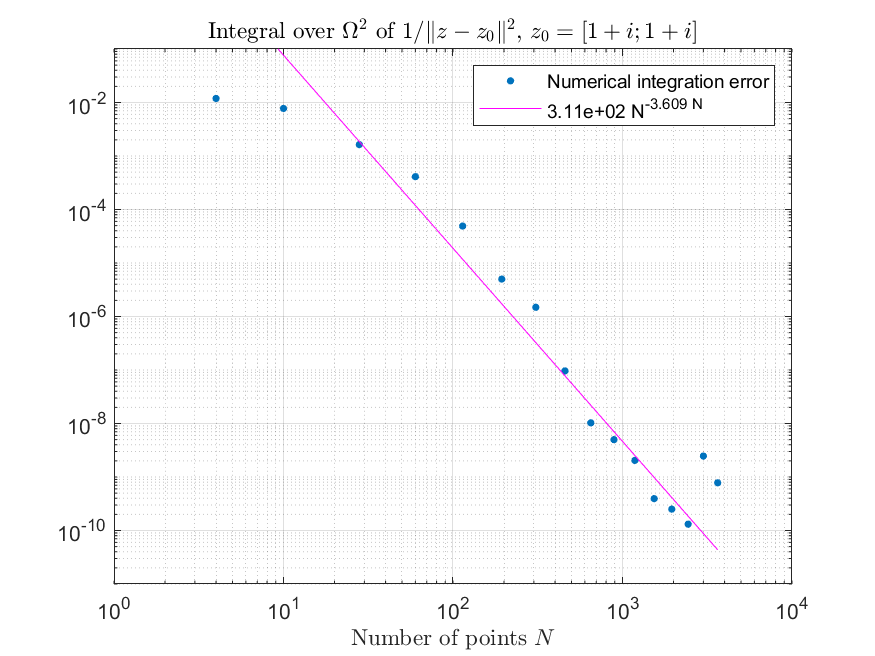}
    \vspace{-4mm}
    \caption{Error of numerical integration vs degree $t$ (left) and number of points $N$ (right) for
    $\bx_0 = (1+\imu,1+\imu)$ on $\csd[2]$ by symmetric triangular complex spherical $t$-designs for odd $t$, $t\leq 31$.}
    \label{fig:intf}
\end{figure}

\begin{table}[ht]
\begin{minipage}{\columnwidth}
\caption{Triangular complex spherical $t$-designs on $\Omega^2$ from real symmetric spherical $t$-designs on $\sph{3}$}
\label{tab:csd2}
\begin{center}
\begin{tabular}{rrrrrr}
\toprule
Degree $t$ & $N$ & $V_{t,N,\varphi_t}$ & Separation & Mesh Norm & Mesh Ratio\\
\midrule
1	& 4	&	0.00e+00	&1.57080 &	1.57080	& 2.0000\\
3 & 8 & -1.73e-17 & 1.57080 & 1.04720 & 1.3333 \\
3	&10	&	-2.58e-16	&1.31812	&0.97760	 &1.4833\\
5	&28	& 9.18e-17	&0.83340	&0.73031 &	1.7526\\
7	&60	&	-1.12e-16	&0.63236	&0.56564 	&1.7890\\
9	&114& -1.40e-16	&0.48633	&0.45479	 & 1.8703\\
11	&194& -2.20e-16	&0.41258	&0.38599	 &1.8711\\
13	&308&	-1.13e-16	&0.34535	&0.32204	 &1.8650\\
15	&458&	-2.74e-17	&0.29137	&0.28774	 &1.9751\\
17	&650&  -7.36e-17	&0.26493	&0.25837	 &1.9505\\
19	&890&	-5.54e-17	&0.23885	&0.23797	 &1.9927\\
21	&1184&	6.63e-17	&0.21389	&0.21131	 &1.9759\\
23	&1538	&-2.50e-16	&0.19990	&0.19515 &1.9524\\
25	&1954	&-2.04e-17	&0.17949	&0.17776	 &1.9808\\
27	&2440	&-8.24e-18	&0.16878	&0.16625	 &1.9701\\
29	&3000	&1.66e-14	&0.15742	&0.15450	 &1.9628\\
31	&3642	&5.06e-15	&0.14612	&0.14621 &2.0013\\
\bottomrule
\end{tabular}
\end{center}
\end{minipage}
\end{table}

\begin{table}[ht]
\begin{minipage}{\columnwidth}
\caption{Triangular complex spherical $t$-designs on $\Omega^3$ from real symmetric spherical $t$-designs on $\sph{5}$}
\label{tab:csd3}
\begin{center}
\begin{tabular}{rrrrrr}
\toprule
Degree $t$ & $N$ & $V_{t,N,\varphi_t}$ & Separation & Mesh Norm & Mesh Ratio\\
\midrule
 3&	   12 & -7.40e-17 & 1.57080	& 1.15026 & 1.4646\\
 5&    56 & -1.89e-16 & 0.93591	& 0.84189 & 1.7991\\
 7&   192 &	-2.15e-17 & 0.65646	& 0.67378 & 2.0528\\
 9&   522 &	-5.50e-17 & 0.51711	& 0.55463 & 2.1451\\
11&	 1208 &	 6.64e-18 & 0.40373	& 0.46572 & 2.3071\\
\bottomrule
\end{tabular}
\end{center}
\end{minipage}
\end{table}

\begin{table}[ht]
\begin{minipage}{\columnwidth}
\caption{Triangular complex spherical $t$-designs on $\Omega^4$ from real symmetric spherical $t$-designs on $\sph{7}$}
\label{tab:csd4}
\begin{center}
\begin{tabular}{rrrrrr}
\toprule
Degree $t$ & $N$ & $V_{t,N,\varphi_t}$ & Separation & Mesh Norm & Mesh Ratio\\
\midrule
3&	16	&	2.22e-16 &	1.57080	&1.20943	 &1.5399\\
5&	102	& -1.40e-18	&0.89818	&0.92098	 &2.0508\\
7&	498	& 2.17e-16	&0.63925	&0.73368	 &2.2954\\
\bottomrule
\end{tabular}
\end{center}
\end{minipage}
\end{table}

\begin{table}[ht]
\begin{minipage}{\columnwidth}
\caption{Triangular complex spherical $t$-designs on $\Omega^5$ from real symmetric spherical $t$-designs on $\sph{9}$}
\label{tab:csd5}
\begin{center}
\begin{tabular}{rrrrrr}
\toprule
Degree $t$ & $N$ & $V_{t,N,\varphi_t}$ & Separation & Mesh Norm & Mesh Ratio\\
\midrule
3&	20 &-3.79e-16 &	1.57080	& 1.24905 &	1.5903\\
5&	170 & 1.03e-16	&0.95756 &0.96636 & 2.0184\\
\bottomrule
\end{tabular}
\end{center}
\end{minipage}
\end{table}

\begin{table}[ht]
\begin{minipage}{\columnwidth}
\caption{Triangular complex spherical $t$-designs on $\Omega^6$ from real symmetric spherical $t$-designs on $\sph{11}$}
\label{tab:csd6}
\begin{center}
\begin{tabular}{rrrrrr}
\toprule
Degree $t$ & $N$ & $V_{t,N,\varphi_t}$ & Separation & Mesh Norm & Mesh Ratio\\
\midrule
3&	24	& 1.62e-17	&1.57080	&1.27795	 &1.6271\\
5&	260	& 9.02e-16	&0.94163	&1.00814	 &2.1413\\
\bottomrule
\end{tabular}
\end{center}
\end{minipage}
\end{table}

\section{Discussion and Conclusion}\label{sec:conclusion}
We propose new uniform designs on the complex unit sphere $\csd$, which are efficient in that they have a low number of points to integrate complex spherical polynomials of a given degree and have good geometric properties as measured by their mesh ratio. Their theoretical properties are summarized, and the associated optimization procedure detailed. We also provide a theoretical justification of the existence of the {\tcsd}s and {\scsd}s with the optimal order number of points.

This exploration led to several interesting theoretical problems. The smallest \emph{coefficient} $c$ such that a spherical $t$-design on $\csd$ exists for all $N \geq c\: t^{2d-1}$ is not known, cf. \eqref{eq:Nhat2} and \eqref{eq:Nbar2}. In particular do triangular complex $t$-designs exist on $\csd$ with $N$ given by \eqref{eq:Nhat1} or \eqref{eq:Nbar1} for all $t$ and $d \geq 2$? A variational characterization of rectangular complex $(k, l)$ designs and associated optimization approach may provide extra freedom, although it is unlikely to provide designs with special symmetry properties. The symmetric {\rsd} approach provides a {\rsd} on the real projective space $\mathbb{R}P^{d}$. A natural question to ask is if it is possible to construct a {\rsd} over more general space with symmetry or quotient structure or even a general manifold. Directly solving the optimization problem \eqref{eq:optproblsd}, which is both nonconvex and nonsmooth, is an intriguing challenge.

\section{Proofs}\label{sec:proof}
In this section we give the detailed proofs of previous results.
\subsection{Proof for Theorem~\ref{thm:existence_complex_design}}
We now prepare some notations for the proof of Theorem~\ref{thm:existence_complex_design}.
 Let $\mathcal{P}_t$ be the Hilbert space of harmonic polynomials on $\Omega^d$ of degree at most $t$ such that 
$$
\int_{\Omega^d} P(\bz) \intc{z} = 0
$$
for all $P\in \mathcal{P}_t$ equipped with the usual inner product
$$
\langle P, Q \rangle := \int_{\Omega^d} P(\bz) Q(\bz) \intc{z}.
$$
By the Riesz representation theorem, for each point $\bz \in \csd$, there exists a unique polynomial $G_{\bz} \in \mathcal{P}_t$ such that
$$
\langle G_{\bz}, Q \rangle = Q(\bz) \quad \mbox{for all $Q \in \mathcal{P}_t$.}
$$
Then, the set of points $\bz_1,\bz_2, \cdots,\bz_N \in \Omega^d$ forms a complex $(k,l)$-design if and only if 
$$
G_{\bz_1} + \cdots + G_{\bz_N} =0.
$$
Let $\nabla$ be the gradient on $\csd$: for $\bz=(z_1,\dots,z_d)$,
\begin{equation*}
    \nabla f(\bz) = (f_1(\bz),\dots,f_d(\bz)),\quad 
    f_k=\frac{\partial f(\bz)}{\partial z_k} -z_k\sum_{j=1}^d z_j \frac{\partial f(\bz)}{\partial z_j},\quad  k=1,\dots,d.
\end{equation*}

Define an open subset $\mathcal{B}$ of the vector space $\mathcal{P}_t$ by
\begin{equation} \label{eq:B}
\mathcal{B} :=\left \{ P \in \mathcal{P}_t \:\Big| \int_{\csd} |\nabla P(\bz)| \intc{z} <1 \right\}.
\end{equation}
Then, the boundary of $\mathcal{B}$ is
\begin{equation*}
\partial\mathcal{B} :=\left \{ P \in \mathcal{P}_t \:\Big| \int_{\csd} |\nabla P(\bz)| \intc{z} =1 \right\}.
\end{equation*}
We define the mapping $L : (\Omega^d)^N \rightarrow \mathcal{P}_t$ by 
$$
(\bz_1,\cdots, \bz_N) \rightarrow G_{\bz_1}+ \cdots + G_{\bz_N}.
$$

The main recipe to prove the main Theorem~\ref{thm:existence_complex_design} is the following lemma about the Brouwer degree theory.
\begin{lemma}[{\cite[Theorems~1.2.6 and 1.2.9]{cho2006topological}}]\label{lem:brower_degree} 
Let $f: \mathbb{R}^n \rightarrow \mathbb{R}^n$ be a continuous mapping and $\mathcal{D}$ an open bounded subset with boundary $\partial \mathcal{D}$ such that $0 \in \mathcal{D} \subset \mathbb{R}^n$. If $\langle x, f(x) \rangle > 0$ for all $ x \in \partial \mathcal{D}$, then there exists $x \in \mathcal{D}$ satisfying $f(x)=0$.
\end{lemma}
Suppose we can find a continuous mapping $F: \mathcal{P}_t \rightarrow (\Omega^d)^N$, where, $F(P):= (\bz_1(P), \cdots, \bz_N(P))$, such that for all $P \in \partial \mathcal{B}$, there holds
$$
\sum_{i=1}^N P (\bz_i(P))>0\,.
$$
When such $F$ exists, if we set the composition mapping 
\begin{equation}\label{Definition f=LcircF}
f:= L \circ F : \mathcal{P}_t \rightarrow \mathcal{P}_t\,, 
\end{equation}
by applying Lemma~\ref{lem:brower_degree} by letting $\mathcal{D}=\mathcal{B}$ and using the identity
$$
\left\langle P, f(P) \right\rangle = \sum_{i=1}^N P(\bz_i(P))
$$
for each $P \in \mathcal{P}_t$,
we would obtain Theorem~\ref{thm:existence_complex_design}. Clearly, the problem boils down to constructing such $F$.

We now explicitly construct function $F$ for each $N \geq C_d t^d$. 
To achieve that,
we extensively use the following notion of an area-regular partition. Let $\mathcal{R}= \{R_1, \cdots, R_N \}$ be a finite collection of closed set $R_i \subset \Omega^d$ such that $\cup_{i=1}^N R_i = \Omega^d$ and $\mu_d(R_i \cap R_j)=0$ for all $ 1\leq i < j \leq N$. The partition norm for $\mathcal{R}$ is defined by 
$$
\|\mathcal {R} \|: = \max_{R \in \mathcal{R}} \hbox{diam} (R), 
$$
where $\hbox{diam}(R):=\max_{\cu,\cv\in R} \distc(\cu,\cv)$ stands for the diameter of the set $R$, which is the geodesic distance of the two most-distant points in $R$.

\begin{lemma}\label{lem:partition}
Let $\Omega^d$ be the $d$-dimensional complex sphere for $d\geq1$. For each $N \in \mathbb{N}$, there exists an area-regular partition $\mathcal{R}=\{R_1,R_2, \cdots, R_N\}$ of $\csd$ with $\|\mathcal{R}\|< B_d N^{-\frac{1}{2d-1}}$, where $B_d$ is a constant which depends only on $d$.
\end{lemma}
\begin{proof}
	By \cite{bourgain1988distribution,kuijlaars1998asymptotics}, there exists an area-regular $N$-partition $\mathcal{R'}$ on $\sph{2d-1}$ with $\|\mathcal{R'}\|< B'_d N^{-\frac{1}{2d-1}}$, where $B'_d$ is a constant depending only on $d$. We use \eqref{eq:map.sph.csd} to map this partition from $\sph{2d-1}$ to $\csd$. By Proposition~\ref{Proposition: csd rsd relationship}, the geometric properties of the partition are preserved, thus, the corresponding $N$-partition $\mathcal{R}$ on $\csd$ satisfies $\|\mathcal{R}\|< B'_d N^{-\frac{1}{2d-1}}$.
\end{proof}

We need to use the following Marcinkiewicz-Zygmund type inequality, which is a special case of \cite[Theorem~5.1]{FiMh2011}.
\begin{lemma}\label{lem:MZineq}
Fix $m\in \mathbb{N}$. For $\eta\in (0,1)$, there exists a constant $r_d$ depending only on $d$ such that for each area-regular partition $\mathcal{R}=\{R_1,\dots,R_N\}$ of $\Omega^d$ with $\|\mathcal{R}\|<r_d/m$, each set of $N$ points $x_i$, with $x_i \in R_i$ for $i = 1, \dots, N$, and each complex spherical polynomial $P$ of total degree $m$ on $\Omega^d$, we have the Marcinkiewicz-Zygmund type inequality 
\begin{equation*}
(1-\eta) \int_{\csd}|P(\bz)| \intc{z} \leq \frac{1}{N} \sum_{i=1}^N | P(\bz_i)| \leq (1+\eta)\int_{\csd}|P(\bz)| \intc{z}.
\end{equation*}
\end{lemma}

With this Lemma, we have the following result.
\begin{corollary}\label{cor:mzineq_grad}
For each area-regular partition $\mathcal{R}=\{R_1, \cdots, R_N\}$ of $\Omega^d$ with $\|\mathcal{R}\| < \frac {r_d} {m+1}$, each collection of $N$ points $x_i$ where $x_i \in R_i$, $i=1, \dots, N$, and each complex spherical polynomial $P\in \mathbb{H}_m(\csd)$, 
\begin{equation*}
\frac 1 {3\sqrt d} \int_{\csd} |\nabla P(\bz)| \intc{z} \leq  \frac{1}{N} \sum_{i=1}^N |\nabla P(\bz_i)| \leq  3\sqrt d \int_{\csd} |\nabla P(\bz)| \intc{z}. 
\end{equation*}
\end{corollary}
\begin{proof}
For a polynomial $P$, the modulus of its gradient at $\bz=(z_1,\dots,z_d)\in \Omega^d$ is
\begin{equation*}
    |\nabla P(\bz)| = \sqrt{\sum_{k=1}^{d}|P_k(\bz)|^2},
\end{equation*}
where for $k=1,\dots,d$,
\begin{equation*}
    P_k(\bz) :=\frac{\partial P(\bz)}{\partial z_k} -z_k\sum_{j=1}^d z_j \frac{\partial P(\bz)}{\partial z_j},
\end{equation*}
which is a polynomial of degree up to $m+1$. In the following, we use the inequality
\begin{equation*}
    \frac{1}{\sqrt{d}}\sum_{k=1}^d|P_k(\bz)|
    \leq |\nabla P(\bz)|
    \leq \sum_{k=1}^d|P_k(\bz)|,\quad x\in \csd,
\end{equation*}
and the Marcinkiewicz-Zygmund type inequality in Lemma~\ref{lem:MZineq}.
Then,
\begin{align*}
    \frac{1}{N}\sum_{i=1}^N|\nabla P(\bz_i)| 
    &\leq \sum_{k=1}^{d}\frac{1}{N}\sum_{i=1}^N|P_k(\bz_i)|\leq \sum_{k=1}^{d}(1+\eta)\int_{\Omega^d}|P_k(\bz)|\intc{z}\\
    &\leq 3\sqrt{d}\int_{\csd}|
    \nabla P(\bz)|\intc{z},
\end{align*}
where we use $\eta=1/2$ in the final inequality.
On the other hand,
\begin{align*}
    \frac{1}{N}\sum_{i=1}^N|\nabla P(\bz_i)| 
    &\geq\frac{1}{\sqrt{d}} \sum_{k=1}^{d}\frac{1}{N}\sum_{i=1}^N|P_k(\bz_i)|\\
    &\geq \frac{1-\eta}{\sqrt{d}}\sum_{k=1}^{d}\int_{\csd}|P_k(\bz)|\intc{z}\geq \frac{1}{3\sqrt{d}}\int_{\csd}|
    \nabla P(\bz)|\intc{z},
\end{align*}
thus completing the proof.
\end{proof}

We are ready to construct $F$ explicitly. For $d,t \in \mathbb{N}$, take 
\begin{equation}
C_d > \left(\frac{54 d B_d}{r_d}\right)^{2d-1},
\end{equation} 
where $B_d$ is the constant in Lemma~\ref{lem:partition} and $r_d$ is defined in Lemma \ref{lem:MZineq},
and fix 
\begin{equation}
N \geq C_d t^{2d-1}. 
\end{equation}
Then, we can construct the mapping $F : \mathcal{P}_t \rightarrow (\Omega^d)^N$ so that $f$ defined in \eqref{Definition f=LcircF} satisfies the assumption of Lemma~\ref{lem:brower_degree}. By Lemma~\ref{lem:partition}, we can take an area-regular partition $\mathcal{R}=\{R_1, \cdots, R_N\}$ with 
$$
\|\mathcal{R}\| \leq B_d N^{-\frac{1}{2d-1}} < \frac {r_d}{54dt}\,.
$$
We then choose an arbitrary $x_i \in R_i$ for each $i = 1, \cdots, N$. Put $\epsilon = \frac{1}{6 \sqrt d}$, and consider the function
\begin{equation}\label{eq:hepsilon}
h_\epsilon(u) := \left \{ 
\begin{aligned}
& u \quad \mbox{ if $u > \epsilon$}, \\
& \epsilon \quad \mbox{ otherwise}.
\end{aligned}
\right.
\end{equation}
Taking a mapping $U : \mathcal{P}_t \times \Omega^d \rightarrow \mathbb{C}^d$ such that
\begin{equation}\label{eq:UPy}
	U(P,\cu) := \frac{\nabla P(\cu)}{h_\epsilon(|\nabla P(\cu)|)}.
\end{equation}
For each $i=1,\cdots, N$, let $\cu_i : \mathcal{P}_t \times [0, \infty) \rightarrow \Omega^d$ be the map satisfying the differential equation 
\begin{equation}\label{eq:dyi}
    \frac{\dd}{\dd s}\cu_i(P,s) = U(P,\cu_i(P,s))
\end{equation}
with the initial condition
$$
	\cu_i(P,0) = \bz_i
$$
for each $P \in \mathcal{P}_t$. Since $U(P,\cu)\leq 1$ for $(P,\cu)\in \mathcal{P}_t \times \csd$, the equation in \eqref{eq:dyi} has a solution. Finally, put
$$
F(P) = (\bz_1(P), \cdots, \bz_N(P)):= \left(\cu_1\left(P,\frac{r_d}{3t}\right), \cdots, \cu_N\left(P,\frac{r_d}{3t}\right)\right)\,.
$$
That is, we construct $F$ by moving the arbitrarily chosen $x_i$ to the ``right place'' via the dynamics \eqref{eq:dyi}.
\begin{lemma}
Let $F : \mathcal{P}_t \rightarrow (\Omega^d)^N$ be the mapping defined by above equation and $\mathcal{B}$ the open set of polynomials given by \eqref{eq:B}. Then, for each $P \in \partial \mathcal{B}$, 
$$
\frac 1 N \sum_{i=1}^N P(\bz_i(P))>0.
$$
\end{lemma}
For each $P \in \partial \mathcal{B}$, by definition we have
$$
\int_{\csd} |\nabla P(\bz)| \intc{z} =1.
$$
To simplify the notation, we write $\cu_i(s)$ in place of $\cu_i(P,s)$. By the Newton-Leibniz formula, we have
\begin{align}
       \frac 1 N \sum_{i=1}^N P(\bz_i(P)) &= \frac 1 N \sum_{i=1}^N P(\cu_i(r_d/3t))\notag \\
       &= \frac 1 N \sum_{i=1}^N P(\bz_i(P)) +\int_0^{r_d/3t} \frac \dd {\dd s} \left[\frac 1 N \sum_{i=1}^N P(\cu_i(s))\right] \dd s.\label{eq:NLformula}
\end{align}
We now need to estimate the quantities
$$
\left|\frac 1 N \sum_{i=1}^N P(\bz_i)\right| \quad \mbox{ and } \quad  \frac \dd {\dd s} \left[\frac 1 N \sum_{i=1}^N P(\cu_i(s))\right].
$$
For each $s \in [0, r_d /3t]$, we have 
\begin{align*}
    \left|\frac 1 N \sum_{i=1}^N P(\bz_i)\right| 
    &= \left|\sum_{i=1}^N \int_{R_i} (P(\bz_i)- P(\bz)) \intc{z}\right| \leq  
    \sum_{i=1}^N \int_{R_i} |P(\bz_i)- P(\bz)| \intc{z} \\
    & \leq \frac{\|\mathcal{R}\|}{N} \sum_{i=1}^N \max_{\bz \in \Omega^d: \distc(\bz,\bz_i) \leq \|\mathcal{R}\|} |\nabla P(\bz)|,
\end{align*}
where $\mathrm{dist}(\bz,\bz_i)$ denotes the geodesic distance between points $\bz$ and $\bz_i$ in $\csd$. Hence, for $\bz^*_i \in \csd$ such that $\mathrm{dist}(\bz^*_i,\bz_i) \leq \|\mathcal{R}\|$ and 
$$
\left|\nabla P(\bz^*_i)\right| = \max_{\bz \in \csd: \distc(\bz,\bz_i) \leq \|\mathcal{R}\|} |\nabla P(\bz)|,
$$
we obtain
$$
\left|\frac 1 N \sum_{i=1}^N P(\bz_i)\right| \leq \frac{\|\mathcal{R}\|}{N} \sum_{i=1}^N |\nabla P(\bz^*_i)|.
$$
Consider the extended area-regular partition $\mathcal{R}' =\{R_1', \cdots, R_N'\}$ defined by $R_i'= R_i \cup \{\bz^*_i\}$. Then, 
\begin{equation}
\|\mathcal{R}'\| \leq 2\|\mathcal{R}\| \mbox{ and }\|\mathcal{R}'\| < r_d / (27dt)\,. 
\end{equation}
Applying Corollary~\ref{cor:mzineq_grad}, we obtain that
\begin{equation}\label{eq:quadrule_UB}
    \left|\frac 1 N \sum_{i=1}^N P(\bz_i)\right| \leq 3\sqrt{d}\: \|\mathcal{R}\| \int_{\csd} |\nabla P(\bz)| \intc{z} < \frac{r_d}{18 \sqrt d t}.
\end{equation}
Now, using \eqref{eq:hepsilon}, by a direct bound, we have
\begin{align*}
    \frac \dd {\dd s} \left[ \frac 1 N \sum_{i=1}^N P(\cu_i(s))\right] &= \frac 1 N \sum_{i=1}^N \frac{|\nabla P(\cu_i(s))|^2}{h_{\epsilon}(|\nabla P(\cu_i(s))|)} \\
    & \geq \frac 1 N \sum_{i: |\nabla P(\cu_i(s))| \geq \epsilon} |\nabla P(\cu_i(s))|   \geq \frac 1 N \sum_{i=1}^N |\nabla P(\cu_i(s))| - \epsilon\,.
\end{align*}
For each $y_i(s) \in \Omega^d$, by \eqref{eq:UPy} and \eqref{eq:dyi},
\begin{equation}
    \left|\frac{\dd \cu_i(s)}{\dd s}\right| \leq \frac{|\nabla P(\cu_i(s))|}{h_{\epsilon}(|\nabla P(\cu_i(s))|)} \leq 1.
\end{equation}
Then, with the mean value theorem,
\begin{equation}\label{eq:d_xi_yi}
    \mathrm{dist}(\bz_i,\cu_i(s))=\mathrm{dist}(\cu_i(0),\cu_i(s)) \leq s.
\end{equation}
For each $s \in [0 , r_d/3t]$, let $\mathcal{R}_i''= R_i \cup \{y_i(s)\}$ and consider the extended area-regular partition $\mathcal{R}''=\{R''_1, \ldots , R''_N\}$. Using \eqref{eq:d_xi_yi}, the norm of $\mathcal{R}''$ is bounded by 
$$ 
\|\mathcal{R}''\| < \frac{r_d}{54}+\frac{r_d}{3t}.
$$
We can then apply Corollary~\ref{cor:mzineq_grad} with $\mathcal{R}''$ to obtain for each $P \in \partial \mathcal{B}$ and $s \in [0, r_d/ 3t]$,
\begin{align*}
\frac \dd {\dd s} \left[ \frac 1 N \sum_{i=1}^N P(\cu_i(s))\right]  &\geq  \frac 1 N \sum_{i=1}^N |\nabla P(\cu_i(s)) |  - \frac{1}{6\sqrt d}  \\
& \geq \frac{1}{3\sqrt d} \int_{\csd} |\nabla P(\bz)| \intc{z} - \frac{1}{ 6 \sqrt d} = \frac{1}{6 \sqrt d},
\end{align*}
where $\epsilon = \frac{1}{6\sqrt{d}}$ is used. This together with \eqref{eq:NLformula} and \eqref{eq:quadrule_UB} give 
$$
\frac{1}{N} \sum_{i=1}^N P(\bz_i(P)) > \frac 1 {6 \sqrt d} \frac{r_d}{3t} - \frac {r_d}{18 \sqrt d t} = 0,
$$
thus completing the proof. 

\subsection{Proofs for Section~\ref{sec:uniformity}}\label{sec:uniformproof}

\begin{proof}[Proof of Proposition~\ref{Proposition: csd rsd relationship}] 
For the separation distance, by \eqref{eq:distc.equiv.distr},
\begin{equation*}
	\sepc_{\phi^{\texttt{R}\to \texttt{C}}(X_N)} = \min_{\stackrel{1\le i,j\le N}{i\neq j}} \distc(\phi^{\texttt{R}\to \texttt{C}}(\bx_i),\phi^{\texttt{R}\to \texttt{C}}(\bx_j)) = \min_{\stackrel{1\le i,j\le N}{i\neq j}} \distr(\bx_i,\bx_j) = \sepr_{X_N}.
\end{equation*}
For the covering distancee, since for any $\by\in\sph{2d-1}$, $\phi^{\texttt{R}\to \texttt{C}}(\by)\in\csd$, then,
	\begin{align}
		\coverc_{\phi^{\texttt{R}\to \texttt{C}}(X_N)} 
		&= \max_{\bz\in\cd} \min_{1\le j\le N}\distc(\bz,\phi^{\texttt{R}\to \texttt{C}}(\bx_j))\notag\\
		&\ge \max_{\by\in\sph{2d-1}} \min_{1\le j\le N}\distc(\phi^{\texttt{R}\to \texttt{C}}(\by),\phi^{\texttt{R}\to \texttt{C}}(\bx_j))\notag\\
		&\ge\max_{\by\in\sph{2d-1}} \min_{1\le j\le N}\distr(\by,\bx_j)=\coverr_{X_N}.\label{eq:coverc.ge.coverr}
	\end{align}
	Since $\phi^{\texttt{R}\to \texttt{C}}$ is a surjective mapping from $\sph{2d-1}$ to $\csd$, for each $\bz\in\csd$, there exists a point $\by_0\in \sph{2d-1}$ such that $\bz=\phi^{\texttt{R}\to \texttt{C}}(\by_0)$. Then, by \eqref{eq:distc.equiv.distr} again,
	\begin{align}\label{eq:coverc.le.coverr.1}
		\min_{1\le j\le N}\distc(\bz,\phi^{\texttt{R}\to \texttt{C}}(\bx_j)) 
		&= \min_{1\le j\le N}\distc(\phi^{\texttt{R}\to \texttt{C}}(\by_0),\phi^{\texttt{R}\to \texttt{C}}(\bx_j))\notag\\
		&= \min_{1\le j\le N}\distr(\by_0,\bx_j).
	\end{align} 
	Taking the maximum of the left hand side of \eqref{eq:coverc.le.coverr.1} over $\bz\in\csd$ gives
	\begin{align*}
		\coverc_{\phi^{\texttt{R}\to \texttt{C}}(X_N)} 
		&= \max_{\bz\in\cd}\min_{1\le j\le N}\distc(\bz,\phi^{\texttt{R}\to \texttt{C}}(\bx_j))\\ 
		&= \max_{\stackrel{\by_0\in\sph{2d-1}}{\bz=\phi^{\texttt{R}\to \texttt{C}}(\by_0)}} \min_{1\le j\le N}\distr(\by_0,\bx_j) \le \coverr_{X_N}.
	\end{align*}
	This and \eqref{eq:coverc.ge.coverr} give $\coverc_{\phi^{\texttt{R}\to \texttt{C}}(X_N)}=\coverr_{X_N}$. We thus complete the proof.
\end{proof}

\section*{Acknowledgement} 
Hau-Tieng Wu acknowledges the hospitality of National Center for Theoretical Sciences (NCTS), Taipei, Taiwan during summer, 2019. 
The authors are grateful for the helpful discussion with Danylo Radchenko on the proof of the existence of the optimal-order complex spherical designs.

{\rm \bibliographystyle{siamplain}
\bibliography{siam_ref}}
\end{document}